\documentclass{cta-author}

\newtheorem{theorem}{Theorem}{}
{}
{}
\newtheorem{lemma}{Lemma}{}
\usepackage{tabularx,cellspace}
\usepackage{amsfonts,amssymb,bm}
\usepackage{multirow,tcolorbox}
\usepackage{tcolorbox}
\usepackage{mathtools}
\usepackage{physics,bm}
\usepackage{cases}
\usepackage{pgfplots, pgfplotstable,subfig}
\usepgfplotslibrary{statistics}
\pgfplotsset{compat=1.15}
\usepackage{interval}
\usepackage{hyperref}
\usepackage[ampersand]{easylist}
\usepackage{cancel}
\usepackage{nomencl,makecell}
\hypersetup{
colorlinks,
citecolor=blue,
filecolor=black,
linkcolor=blue,
urlcolor=black}
\usetikzlibrary{matrix}
\usetikzlibrary{shapes.multipart, positioning,calc,matrix}
\usetikzlibrary{patterns,arrows,decorations.pathreplacing}
\usetikzlibrary{shapes,arrows}
\usepackage{algorithm}
\usepackage{algpseudocode}
\usepackage{standalone}
\usepackage{makecell,bm}
\usepackage{tikz,amsmath, amssymb,bm,color}
\usepackage{xcolor}
\usetikzlibrary{shapes,arrows.meta,calc,fit,backgrounds,shapes.multipart,positioning}
\usetikzlibrary{calc}
\usepackage{makecell}
\usetikzlibrary{backgrounds}
\usetikzlibrary{arrows.meta}
\usepackage{varwidth}

\tikzstyle{startstop} = [rectangle, rounded corners, minimum width=1cm, minimum height=1cm,text centered, draw=black, fill=black!00] 
\tikzstyle{i} = [trapezium, trapezium left angle=70, trapezium right angle=110, minimum width=3cm, minimum height=1cm, text centered, text width=6cm, draw=black, fill=red!30]
\tikzstyle{o} = [trapezium, trapezium left angle=70, trapezium right angle=110, minimum width=3cm, minimum height=1cm, text centered, text width=4.5cm, draw=black, fill=green!30]
\tikzstyle{process} = [rectangle, minimum width=1cm, minimum height=1cm,text centered, text width=1cm ,  draw=blue, fill=black!00]
\tikzstyle{decision} = [diamond, minimum width=3cm, minimum height=0.5cm, text centered, inner sep=1pt, draw=blue, fill=violet!00]
\tikzstyle{joint} = [draw=black,circle,text centered,minimum width=1cm, minimum height=1cm,fill=yellow!20]
\tikzstyle{arrow} = [->,>=stealth,draw= blue]

\tikzset{
connector/.style = {draw,circle,minimum width=1cm, minimum height=1cm, text centered,fill=yellow!20}
}
\tikzset{block/.style={rectangle split, draw, rectangle split parts=2, text badly centered, text width = 4cm, font=\fontsize{10}{0}\selectfont},   
line/.style={draw, -{Latex[length=2mm,width=1mm]}},
cloud/.style={draw, ellipse,fill=white!20, node distance=3cm,    minimum height=4em},  
container/.style={draw, rectangle,dashed,inner sep=0.28cm, rounded
corners,fill=yellow!20,minimum height=1.2cm}}

\newdimen\XCoord
\newdimen\YCoord

\pgfdeclarelayer{background}
\pgfdeclarelayer{foreground}
\pgfsetlayers{background,main,foreground}

\usepackage{physics}
\newcommand{\DistFlow}{DistFlow\ }
\newcommand{\diag}{\operatorname{\mathbf{diag}}}

\newcommand{\lyapnouv}{Lyapunov\ }
\newcommand{\lipschitz}{Lipschitz\ }

\makenomenclature
\begin{document}

\supertitle{Research Article}

\title{Lyapunov Stability of Smart Inverters Using Linearized DistFlow Approximation}

\author{\au{Shammya Shananda Saha$^{1,*}$} \au{Daniel Arnold$^{2}$} \au{Anna Scaglione$^1$} \au{Eran Schweitzer$^1$} \au{Ciaran Roberts$^2$} \au{Sean Peisert$^2$} \au{Nathan G. Johnson$^3$} }

\address{\add{1}{School of Electrical, Computer and Energy Engineering, Arizona State University, Tempe, Arizona.}
\add{2}{Lawrence Berkeley National Lab, Berkeley, California.}
\add{3}{The Polytechnic School, Arizona State University, Mesa, Arizona.}
\email{shammya.saha@asu.edu}}

\begin{abstract}
\looseness=-1 \textcolor{blue}{Fast-acting smart inverters that utilize preset operating conditions to determine real and reactive power injection/consumption can create voltage instabilities (over-voltage, voltage oscillations and more) in an electrical distribution network if set-points are not properly configured.} In this work, linear distribution power flow equations and droop-based Volt-Var and Volt-Watt control curves are used to analytically derive a stability criterion using \lyapnouv analysis that includes the network operating condition. The methodology is generally applicable for control curves that can be represented as \lipschitz functions. The derived \lipschitz constants account for smart inverter hardware limitations for reactive power generation. A local policy is derived from the stability criterion that allows inverters to adapt their control curves by monitoring only local voltage, thus avoiding centralized control or information sharing with other inverters. \textcolor{blue}{The criterion is independent of the internal time-delays of smart inverters.} Simulation results for inverters with and without the proposed stabilization technique demonstrate how smart inverters can mitigate voltage oscillations locally and mitigate real and reactive power flow disturbances at the substation under multiple scenarios. The study concludes with illustrations of how the control policy can dampen oscillations caused by solar intermittency and cyber-attacks.
\end{abstract}

\maketitle
\section*{\textcolor{blue}{Nomenclature}}
\begin{table}[htbp]
\resizebox{\columnwidth}{!}{%
{\begin{tabular*}{20pc}{@{\extracolsep{\fill}}lll@{}}
$\mathcal{L}$ & Set of all lines\\ 
$\mathcal{N}$ & Set of all nodes\\ 
$p_i^c$, $p_i^g$ & \makecell[l]{Real power demand, Real power generation at \\ node $i$} \\ 
$q_i^c$, $q_i^g$ & \makecell[l]{Reactive power demand, Real power generation \\ at node $i$}\\ 
$r_{ij},x_{ij}$ & \makecell[l]{Line resistance, Line reactance of line between \\ nodes $i$ and $j$}\\ 
$P_{ij},Q_{ij}$ & \makecell[l]{Real power flow, Reactive power flow of line between \\ nodes $i$ and $j$}\\ 
$v_i$ & Voltage of node $i$\\ 
$c_{ij}$ & Current of line between nodes $i$ and $j$\\ 
$\textcolor{blue}{s_{i}}$ & \textcolor{blue}{ Rated apparent power of inverter at node $i$} \\
$\textcolor{blue}{\overline{p_{i}}}$ & \textcolor{blue}{Maximum real power output of inverter at node $i$ }\\
\textcolor{blue}{$q_i^{lim}$} & \textcolor{blue}{\makecell[l]{Hardware limit of reactive power generation of inverter \\ at node $i$}} \\ 
\textcolor{blue}{$f_{p,i}$} & \textcolor{blue}{Volt-Watt control function of inverter at node $i$} \\ 
\textcolor{blue}{$f_{q,i}$} & \textcolor{blue}{Volt-Var control function of inverter at node $i$} \\ 
\textcolor{blue}{$C_{p,i}$} & \textcolor{blue}{\makecell[l]{\lipschitz constant for Volt-Watt control function  of \\ inverter at node $i$}} \\ 
\textcolor{blue}{$C_{q,i}$} & \textcolor{blue}{\makecell[l]{\lipschitz constant for Volt-Var control function of \\ inverter at node $i$} }\\ 
\end{tabular*}}{}
}
\end{table}

\section{Introduction}
\label{sec:introduction}
\subsection{\textcolor{blue}{Motivation}}
\textcolor{blue}{The fast-acting and complex control mechanisms of Distributed Energy Resources (DER) pose challenges for planning, operations, and reliability of electric distribution systems and microgrids with increasing amounts of DER. Poor DER coordination can create reliability issues even at low penetration levels \cite{KHALILI201992}. At higher penetrations, the net impact of many solar photovoltaic (PV) generators may accumulate and further affect power quality \cite{Hu2019,Singhal2019}.} Power quality control devices could be added to the distribution network to counteract the problem, but the additional cost can be avoided if the problem can be resolved by the solar inverters. One such approach involves using PV inverter capacity to control power (real and reactive) generation/consumption as a means to stabilize voltage swings occurring from changes in load, generation, equipment failure, or cyber-attack. This functionality can be implemented within the allowable scope of interconnection standards \cite{IEEE_1547,rule21} using only local information at the inverter without the need for additional data monitoring and control equipment to be installed on the network. 

\subsection{\textcolor{blue}{Related Literature}}
\label{sec:background}
Numerous works have studied power quality issues for inverter-dominated distribution networks and recommended mitigation strategies for over-voltage, voltage oscillation, and other issues. Methods to perform voltage control by PV inverters can be categorized by the control approach (centralized or local), control function (Volt-Var, combined Volt-Var, and Volt-Watt, none) \cite{Vijayan2019}, and level of detail to which inverter hardware limits are expressed in the control function.

In a centralized approach, a utility or aggregator collects real-time information about the network, processes that information using an optimization formulation (commonly an Optimal Power Flow (OPF) formulation), and sends updated set-points back to the inverters. Such an approach was taken in \cite{bakerNetwork2017} to determine settings of inverters that ensured voltage stability across the entire distribution network, but that approach required the centralized utility to exchange information with inverters and this restricted updates to every 5-15 minutes. Similar approaches could be suitable for smaller networks with fewer inverters and more frequent communication. The authors in \cite{Weckx} provided a centralized optimization formulation with set-points updated at similar timescales to optimize local control curves modeled using a first-order spline model. This approach was well-suited for managing steady-state voltage levels but posed challenges to resolving fast-acting voltage oscillations occurring from cloud intermittency \cite{zhu2016fast}. \textcolor{blue}{By incorporating load tap changers and smart inverters, authors in \cite{Long2019} developed a centralized control method for Volt-Var optimization using sensitivity factors that assumed a constant slope for the Volt-Var controllers of the smart inverters. This approach required recalculation of sensitivity factors by the utility and extensive communication to the smart inverters.} \textcolor{blue}{The issues created by communication delay in a centralized control system or an equivalent master-slave control system are also described in \cite{Muthukaruppan2020}. A master-slave Volt-Var optimization method was presented in \cite{Shi2019} that minimized real power losses and grouped inverters by spatial distance but did not incorporate the Volt-Var control curves of the smart inverters. Another centralized approach presented in}  \cite{GHASEMI2016313} showed how to manage voltage by curtailing active power but had the drawback of requiring simultaneous control of all inverters, a challenge given the time delay in communication and inverter control action. The issues associated with communication and synchronization can be avoided using local control which reduces information exchange and is faster to execute. Among local or decentralized approaches, droop-based voltage control \cite{IEEE_1547,rule21} is the most common framework from literature and utility practices \cite{farivar2013equilibrium,Lipschitz1,Lipschitz2,Lipschitz3,Singhal2019,jahangiri2014distributed,zhu2016fast,Helou2020}. However, improper selection of control parameters can lead to control instability and voltage oscillation \cite{farivar2013equilibrium,jahangiri2014distributed}. Work in \cite{jahangiri2014distributed} highlighted instability concerns and proposed a "delayed droop control" to maintain voltages within acceptable bounds by absorbing or supplying reactive power. But as pointed out in \cite{Singhal2019}, the method in \cite{jahangiri2014distributed} cannot adapt to changing operating conditions and external disturbances that may affect control algorithm convergence. The challenge of convergence was addressed in \cite{Singhal2019} but no generalized method was proposed to choose correction factors ($k_i^d;~\Delta v_{f}$) to manage the shifting of the control curves during unstable operating conditions. Results presented in \cite{Pukhrem} showed that a limiter algorithm can mitigate over-voltage issues, but no analytical model was provided to ensure that the algorithm could address voltage oscillations. Authors in \cite{Braslavsky_stability} conducted an analytical stability analysis for a single inverter connected to a strong source (infinite bus/substation) with potential extensions to an inverter dominated distribution network. A related study completed stability analysis for a distribution network \cite{Heidari2018} using an adaptive decentralized control scheme by assuming an equal ratio of reactance to resistance for all branches and that inverters at load buses could not be controlled. Such assumptions cannot be generally applied to any distribution network.

Inverter control functions are another area of study. Authors in \cite{Pompodakis2016} provided reactive power correction methods to reduce over-voltage issues within a low voltage radial grid without incorporating any Volt-Var or Volt-Watt droop control. Research in Volt-Var control has shown ways to achieve voltage stability for smart inverters \cite{Lipschitz1, Lipschitz2, Lipschitz3,zhu2016fast, Singhal2019,jahangiri2014distributed}. More recently, interconnection standards \cite{IEEE_1547,rule21} require smart inverters to provide both Volt-Var and Volt-Watt control with ongoing research focusing on how to employ the combined approach to provide voltage support and improve reliability \cite{bakerNetwork2017, Heidari2018, GHASEMI2016313, Olivier2016, Aminul2015}. 

The literature review also indicated that models of inverter capability often assume the smart inverter can operate anywhere in the power circle \cite{bakerNetwork2017,farivar2013equilibrium,Lipschitz1,Lipschitz2,Lipschitz3,jahangiri2014distributed,zhu2016fast,Helou2020,Singhal2019}. This allows an inverter to set reactive power generation equal to the apparent power rating of the inverter when there is no real power generation. Such behavior violates hardware limits on maximum reactive power generation that can be lower than the apparent power rating of the inverter \cite{Hu2019,SolarEdge,Sarfaraz2016,Yuan2007}.  

\subsection{Contribution}
\label{contribution}
In this work, piece-wise linear droop models of Volt-Var and Volt-Watt control functions are used to describe inverter behavior where the functions are expressed as \lipschitz functions. Unlike the modeling techniques presented in existing works, this study includes limitations on reactive power generation that account for physical hardware while deriving the \lipschitz constants. \textcolor{blue}{The formulation and derived constants are applicable for all smart inverters with piece-wise linear droop models, and can be used to derive constants for an inverter with reactive power generation capability equal to the apparent power rating, a simplifying assumption common to approaches in prior literature.}

This work also analytically derives stability policies and provides an approach to reduce voltage oscillations and stabilize the network-wide voltage profile using only local information and control actions by using \lyapnouv analysis. \textcolor{blue}{Whereas traditional eigenvalue based stability analysis only can indicate stability based on the sign of the eigenvalues, \lyapnouv analysis allows the development of a robust control policy using the derived stability criterion. In this work,} the network is described using a linear distribution power flow using $v^2$ instead of $v$, thus avoiding the assumptions made in \cite{bakerNetwork2017,zhu2016fast,farivar2013equilibrium, Lipschitz1}. Performing the linearization using $v^2$ ties the network operating condition explicitly to the stability condition unlike work found in \cite{bakerNetwork2017, Singhal2019, Lipschitz1}. The analysis presented holds true as long as droop control curve functions can be represented as \lipschitz functions. Furthermore, by exposing mechanisms that contribute to system instabilities, the paper introduces a completely decentralized policy for restoring a stable voltage profile. Finally, the numerical section showcases how the stability policy can mitigate oscillations induced by operating conditions or cyber-attacks.

\textcolor{blue}{
\subsection{Paper Organization}
Section \ref{sec:distflow} provides the underlying mathematical description of the distribution network using a Matrix-based formulation of the linearized \DistFlow equations from \cite{baran1989optimal,barenwunetworkreconfig}. Section \ref{sec.smartinverter} follows to describe smart inverter operational logic and control functions with a derivation of \lipschitz constants. \lyapnouv analysis is used to derive a sufficient condition to ensure network-wide voltage stability. Section \ref{sec:stability_analysis} provides proof of the stability criterion and the local control model developed using that stability criterion. Simulation results in Section \ref{simulation_results} demonstrate how the approach resolves voltage oscillations from generation intermittency and cyber-attacks, with Section \ref{sec:discussion} concluding the paper with an overall discussion of the study and possible future work.}
\section{Linearized Branch Flow Equations}
\label{sec:distflow}
 \textcolor{blue}{Assuming a balanced, radial system; for every line $(i,j) \in \mathcal{L}$ the \DistFlow equations are \cite{baran1989optimal,barenwunetworkreconfig}:}
\begin{subequations} 
\label{eq:original_equations}
\begin{align}
P_{ij}&=p^c_{j}-p^g_{j}+r_{ij}c^2_{ij}+\sum_{k:(j,k)\in \mathcal{L}} P_{jk} \label{eq:pdistflow} \\
Q_{ij}&=q^c_{j}-q^g_{j}+x_{ij}c^2_{ij}+\sum_{k:(j,k)\in \mathcal{L}} Q_{jk} \label{eq:qdistflow} \\
v_j^2-v_i^2&=-2(r_{ij}P_{ij}+x_{ij}Q_{ij})+(r_{ij}^2+x_{ij}^2)c^2_{ij}\label{eq:vj2-vi2}
\end{align}
\end{subequations}
Representing vectors and matrices by boldface letters, the following vectors and matrices are defined:
\begin{subequations}
\label{eq:vectors}
\begin{align}
\bm{v}_o^2&=(v_0^2,\dots,v^2_{n-1})^T,~~~\bm{v}^2=(v_1^2,\dots,v^2_{n-1})^T  \\
\bm{p}&=(p_0,\dots,p_{n-1})^T,~~~\bm{q}=(q_0,\dots,q_{n-1})^T  \\
\bm{P}&=(P_{01},\dots)^T,~~~\bm{Q}=(Q_{01},\dots)^T \\
\bm{c}^2&=(c^2_{01},\dots)^T,~~\bm{r}=(r_{01},\dots)^T,~~\bm{x}=(x_{01},\dots)^T 
\end{align}
\begin{align}
\label{eq:todefn}
[\bm{T}]_{e,i}&=\begin{cases}
1 & i=t(e)\\
0 & \textnormal{else}
\end{cases}
~~\forall~e\in \mathcal{L}\\
[\bm{F}]_{e,i}&=\begin{cases}
1 & i=f(e)\\
0 & \textnormal{else}
\end{cases}
~~\forall~e\in \mathcal{L}\\
\label{eq:Modefn}
\bm{M}_o&=\bm{F}-\bm{T}
\end{align}
\end{subequations}
Here, $\bm{T}$ and $\bm{F}$ \textcolor{blue}{are referred as the} \textit{to} and \textit{from} matrices, respectively, and $\bm{M}_o$ represents the incidence matrix. Each are of size $ \abs{\mathcal{L}} \times \abs{\mathcal{N}} =(n-1)\times n$. Also, $t(e) \textnormal{and} f(e)$ are two functions that return the \emph{to} and \emph{from} node of an edge, respectively.

Note that $\bm{M}_o\bm{1}=0$, and using the reference voltage $v_0$~(generally the voltage at the root of the feeder) a unique solution for the \DistFlow equations can be found. Let $\bm{M}$ be the $(n-1)\times (n-1)$ matrix obtained removing the first column from $\bm{M}_o$, which is also invertible \cite{zhu2016fast}. Therefore: 
$$\bm{M}_o(\bm{v}_o^2-v_0^2\bm{1}_{n})\equiv \bm{M}(\bm{v}^2-v_0^2\bm{1}_{n-1})$$
Introducing the notation $\diag(\bm{x})$ to indicate a diagonal matrix with the entries of vector $\bm{x}$ in its diagonal entries, the following matrices are defined:
\begin{subequations}
\label{eq:RXL}
\begin{align}
    \bm{R}&=2\bm{M}^{-1}\diag(\bm{r})(\bm{I}-\bm{T}\bm{F}^T)^{-1}\bm{T} \label{eq:R}\\
    \bm{X}&=2\bm{M}^{-1}\diag(\bm{x})(\bm{I}-\bm{T}\bm{F}^T)^{-1}\bm{T} \label{eq:X}\\
    \bm{L}&=\bm{M}^{-1}\left[2\diag(\bm{r})(\bm{I}\!-\!\bm{T}\bm{F}^T)^{-1}\diag(\bm{r})\right.\label{eq:L}\\
            &\left. +2\diag(\bm{x})(\bm{I}\!-\!\bm{T}\bm{F}^T)^{-1}\diag(\bm{x}) -\diag(\bm{r}^2+\bm{x}^2) \right] \nonumber
\end{align}
\end{subequations}
The absolute diagonal values of $\bm{R}$ and $\bm{X}$ matrices represent the electrical distance of the \emph{to} node of a branch from the substation node. Matrix $\bm{L}$ is associated with the ohmic losses of the network.

Using definitions in \eqref{eq:vectors} and \eqref{eq:RXL} and solving for $\bm{P}$ and $\bm{Q}$ permits the \DistFlow equations from \eqref{eq:original_equations} to be written in terms of $\bm{v}^2$:
\begin{align}
    \label{eq.v2}
    \bm{v}^2&=v_0^2\bm{1}+\bm{R}\bm{p}+\bm{X}\bm{q}+\bm{L}\bm{c}^2\\
    \label{eq.c2}
    \bm{c}^2&=\diag^{-1}(\bm{F}\bm{v}^2)(\bm{P}^2+\bm{Q}^2)
\end{align}
Ignoring the losses associated with the term $\bm{L}\bm{c}^2$ as shown in \cite{barenwunetworkreconfig, farivar2013equilibrium}, \eqref{eq.v2} becomes linear in terms of $\bm{v}^2$ and can be written as follows:
\begin{align}
    \label{eq.v3}
    \bm{v}^2=&v_0^2\bm{1}+\bm{R}(\bm{p}^c-\bm{p}^g)+\bm{X}(\bm{q}^c-\bm{q}^g)
\end{align}
Here, the injections in load and generation are partitioned  as $\bm{p}=\bm{p}^c-\bm{p}^g$ and $\bm{q}=\bm{q}^c-\bm{q}^g$, respectively. Hence, \eqref{eq.v3} can be written as follows:
\begin{equation}
    \label{eq:main}
    \bm{v}^2=\underbrace{v_o^2\bm{1}+\bm{Z}\bm{s}^c}_{\bar{\bm{v}}^2}-\bm{Z}\bm{s}^g 
\end{equation} 
where:
\begin{equation}
    \bm{Z}=\mqty[\bm{R} & \bm{X}],~~~\bm{s}^c= \mqty[\bm{p}^c \\ \bm{q}^c],~~~\bm{s}^g= \mqty[\bm{p}^g \\ \bm{q}^g] \label{eq:ZS}
\end{equation}

\section{Smart Inverter Models}\label{sec.smartinverter}
\subsection{Overview of the inverter logic design problem}
The operational logic of smart inverters is aimed at selecting real and reactive power generation set-point values and can be expressed \textcolor{blue}{through a cost function. The cost function $\Gamma$} is a linear combination of two objectives \cite{Sulc_Turitsyn} and whose minimizer is the optimal dispatch for the inverters:
\begin{align}
    \label{eq:cost_function}
    \min_{\bm{v},\bm{p},\bm{q}} \Gamma (\bm{v},\bm{p},\bm{q},\bm{R}, \bm{X})~\mathrm{s.t.}~ (\bm{v},\bm{p},\bm{q})\in \mathcal{S}
\end{align}
where $\mathcal{S}$ is the feasible set for $\bm{v},\bm{p},\bm{q}$ and:
\begin{equation}
\label{eq:cost_optimization}
\Gamma (\bm{v},\bm{p},\bm{q},\bm{R}, \bm{X})=
c_\rho\ \underbrace{ \rho (\bm{p},\bm{q},\bm{R},\bm{X})}_{M1} + c_\nu\ \underbrace{\nu(\bm{v},\bm{p},\bm{q})}_{M2}   
\end{equation}
with $\rho(\bm{p},\bm{q},\bm{R}, \bm{X})$ expressing ohmic losses and $\nu(\bm{v},\bm{p},\bm{q})$ promoting a voltage profile that deviates as little as possible from the nominal voltage. Non-negative weights $c_\rho$ and $c_\nu$ can be tuned to adjust the relative importance of each objective. \\

The feasible set in $\bm{v},\bm{p},\bm{q}$ is defined by the \DistFlow equations. When exclusively minimizing losses (term $M1$ in \eqref{eq:cost_optimization}), the minimization function and constraints can be written as: 
\begin{subequations}
\label{eq:M1}
\begin{align}
     \min_{\bm{p}_g,\bm{q}_g}& (\bm{p}^c\!-\!\bm{p}^g)^T \bm{R} (\bm{p}^c\!-\!\bm{p}^g) \!+\! (\bm{q}^c\!-\!\bm{q}^g)^T \bm{X} (\bm{q}^c-\bm{q}^g) \\
      \textnormal{s.t.} 
     ~~& \bm{p}^g_{\min} \leq \bm{p}^g \leq \bm{p}^g_{\max}; \ \   \bm{q}^g_{\min} \leq \bm{q}^g \leq \bm{q}^g_{\max}
\end{align}
\end{subequations}
Here, $\bm{p}^g_{\min}, \bm{p}^g_{\max}, \bm{q}^q_{\min}, \bm{q}^g_{\max}$ represent the minimum real power, maximum real power, minimum reactive power, and maximum reactive power that an inverter can generate at any instant of time, respectively. This objective does not include any constraint on the voltage \textcolor{blue}{profile; rather, the voltage profile will be the solution of the \DistFlow equations and may violate the standard acceptable} voltage range. 

For the policy focusing only on minimizing voltage deviation from the nominal value (term $M2$ in \eqref{eq:cost_optimization}), the minimization function and constraints can be written as:%
\begin{subequations}%
\begin{align}\label{eq:M2}
     \min_{\bm{p}^g,\bm{q}^g}& \  \norm{\bm{v}-\bm{v}_{nom}}^2_{2} \\
     \textnormal{s.t.} 
     ~~& \bm{v}^2=v_0^2\bm{1}+\bm{R}(\bm{p}^c-\bm{p}^g)+\bm{X}(\bm{q}^c-\bm{q}^g) \\
    & \bm{v}_{\min} \leq \bm{v} \leq \bm{v}_{\max} \\
    & \bm{p}^g_{\min} \leq \bm{p}^g \leq \bm{p}^g_{\max}; \ \   \bm{q}^g_{\min} \leq \bm{q}^g \leq \bm{q}^g_{\max} 
\end{align}
\end{subequations}
where $\bm{v}_{nom}$ represents the nominal voltage vector. This includes additional constraints beyond \eqref{eq:M1}.

The direct minimization of $\Gamma (\bm{v},\bm{p},\bm{q},\bm{R}, \bm{X})$ requires communicating real-time information of all loads and inverters to a centralized solver. The set-points (i.e. the solution of problem \eqref{eq:M2}) can be updated at regular predefined times based on the load components $\bm{p}_c$ and $\bm{q}_c$ and the maximum apparent power that can be generated at a current time. Conversely, a purely decentralized method uses local measurements as inputs for deciding inverter set-points. The control policy is instead hardwired at the time of installation, noting, however, that remote firmware updates are possible but excluded from consideration here. Two forms of policies exist for local control: (1) policies that adjust reactive and/or real power to minimize losses, $M1$ \cite{kernel_based_learning}, and (2) policies that respond to voltage measurements and try reduce voltage deviations by adjusting reactive and/or real power, $M2$. The first policy will not cause oscillations but may result in unacceptable voltage levels. The second policy leads to a closed-loop system and may cause oscillations while trying to maintain voltage within operating limits  \cite{farivar2013equilibrium,Lipschitz1,Lipschitz2,Lipschitz3,Braslavsky_stability} as the \DistFlow equations constrain the values $(\bm{v},\bm{p},\bm{q})$. 

\subsection{Inverter modeling assumptions}
\label{sec:inverter_modeling}
\textcolor{blue}{This work focuses} on the second type of strategy ($M2$) described above while trying to adjust both real and reactive power. This strategy requires defining the Volt-Watt and Volt-Var control functions, respectively, as,
\begin{equation}\label{eq:fp_fq_generic}
   \bm{f_p}(\Delta \bm{v}) : \mathbb{R}^n \mapsto \mathbb{R}^{n} \quad\text{and}\quad 
   \bm{f_q}(\Delta \bm{v}) : \mathbb{R}^n \mapsto \mathbb{R}^{n}
\end{equation}  
\textcolor{blue}{For local policies, these functions take each entry of the vector that represents the deviation from the nominal voltage according to \eqref{eq:dif_voltage} as input, and provide the real and reactive power injection values for the corresponding bus as output.}
\begin{equation}
    \label{eq:dif_voltage}
    \Delta \bm{v}= \bm{v}-\bm{v}_{nom}    
\end{equation}

\textcolor{blue}{Considering $\Delta v_i$ as the voltage deviation of the $i^{th}$ bus,} and $f_{p,i}(\Delta v_i)$ and $f_{q,i}(\Delta v_i)$ representing the Volt-Watt and Volt-Var control functions of the inverter at $i^{th}$ bus, respectively, the following assumptions can be made regarding the control policy \cite{farivar2013equilibrium,Lipschitz1,Lipschitz2,Lipschitz3}:
\begin{description}
     \item[A1: ] Both $f_{p,i}(\Delta v_{i})$ and $f_{q,i}(\Delta v_{i})$ are monotonically decreasing functions and are continuous and piece-wise differentiable (control functions may include regions where the derivative is zero, such as a dead-band or constant output). 
     \item[A2: ] The derivatives of the control functions are bounded, i.e. there exists $C_{p,i}<+\infty$ and $C_{q,i}<+\infty$ such that $\abs{f_p^{'} (\Delta v_i)} \leq C_{p,i} $ and $\abs{f_q^{'} (\Delta v_i)} \leq C_{q,i} $ for all $\Delta v_i$.
\end{description}
For executing Volt-Watt and Volt-Var control, \textcolor{blue}{this work adopts the droop control curves mentioned in \cite{inverter2016,IEEE_1547,rule21}. In these works,} the modulation of active and reactive power injection indicated in \eqref{eq:fp_fq_generic} depends upon \eqref{eq:dif_voltage}. The derivation of the corresponding \lipschitz constants are hereafter shown. This example is used as a case study in the numerical simulations in Section \ref{simulation_results} that corroborate \textcolor{blue}{the} analysis.

\subsection{Definitions}
For an inverter at node $i$ of a distribution network, $s_{i}$ and $\overline{p_{i}}$ represent the rated apparent power and the maximum real power output at certain irradiance respectively. Following \cite{Braslavsky_stability} and dropping the superscript $^g$ representing generation, $\overline{p_{i}}$ can be expressed as a fraction of $s_{i}$: 
\begin{equation}
\label{eq:mu}
  \overline{p_{i}} = \mu\ s_{i}; ~~~ 0 < \mu \leq 1  
\end{equation}
At nominal irradiance, the inverter can generate real power equal to its apparent power rating, resulting in $\mu = 1$. \\
The maximum reactive power consumption/injection depends on hardware limits ($q_{i}^{lim}$) and available reactive power. Hence:
\begin{equation}
\label{eq:qbar_v}
    \overline{q_{i}} (\Delta v_i) = \min \left(q_{i}^{ lim},\sqrt{s_{i}^2 - f_{p,i}^2 (\Delta v_i)}\right)
\end{equation} 
For constant values of $f_{p,i} (\Delta v_i)$, $\overline{q_{i}} (\Delta v_i)$ is independent of voltage deviation and can be written as:
\begin{equation}
\label{eq:qbar}
 	\overline{q_{i}} (\Delta v_i) = \min \left(q_{i}^{ lim},\sqrt{s_{i}^2 - f_{p,i}^2}\right) = \overline{q_{i}}   
\end{equation}
The piece-wise linear Volt-Watt and Volt-Var control curves (referred to as \emph{droop} curves) of the inverter at the $i^{th}$ bus respectively are shown in Figures \ref{fig:vwc_new} and \ref{fig:vvc_new} where, $\epsilon_p,~\epsilon_q^{+},~\epsilon_q^{-},~V_p\ \textnormal{and}\ V_q^{+}$ are positive, $V_q^{-}$ is negative, $V_q^{+}-\epsilon_q^{+}/2>0$, $V_q^{-}+\epsilon_q^{-}/2<0$,  $V_p-\epsilon_p/2>0$. \textcolor{blue}{Using  \eqref{eq:mu}-\eqref{eq:qbar}, \eqref{eq:vwcsg} and \eqref{eq:vvcsqg}} provide the mathematical formulation representing Figures \ref{fig:vwc_new} and \ref{fig:vvc_new}, respectively.
\begin{figure}[htbp]
	\centering
	\includegraphics[width=1\columnwidth]{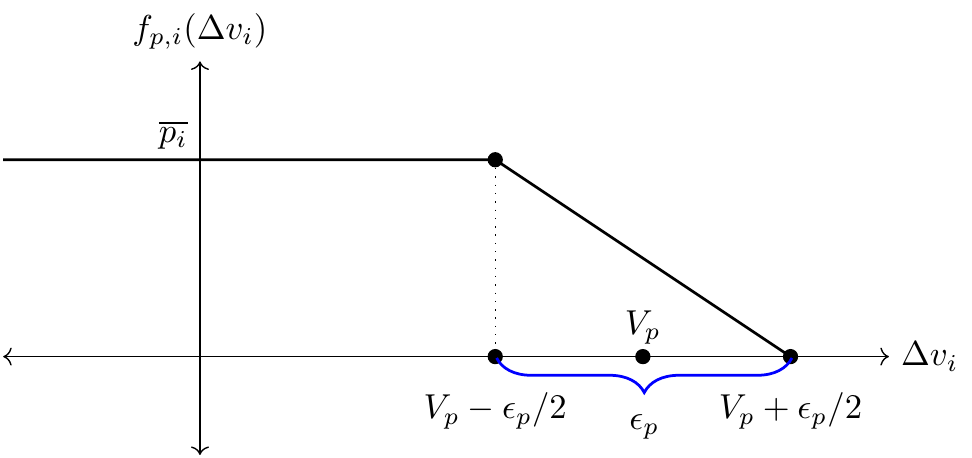}
	\caption{Inverter Volt-Watt Curve }
	\label{fig:vwc_new}
\end{figure}
\begin{figure}[htbp]
	\centering
	\includegraphics[width=1\columnwidth]{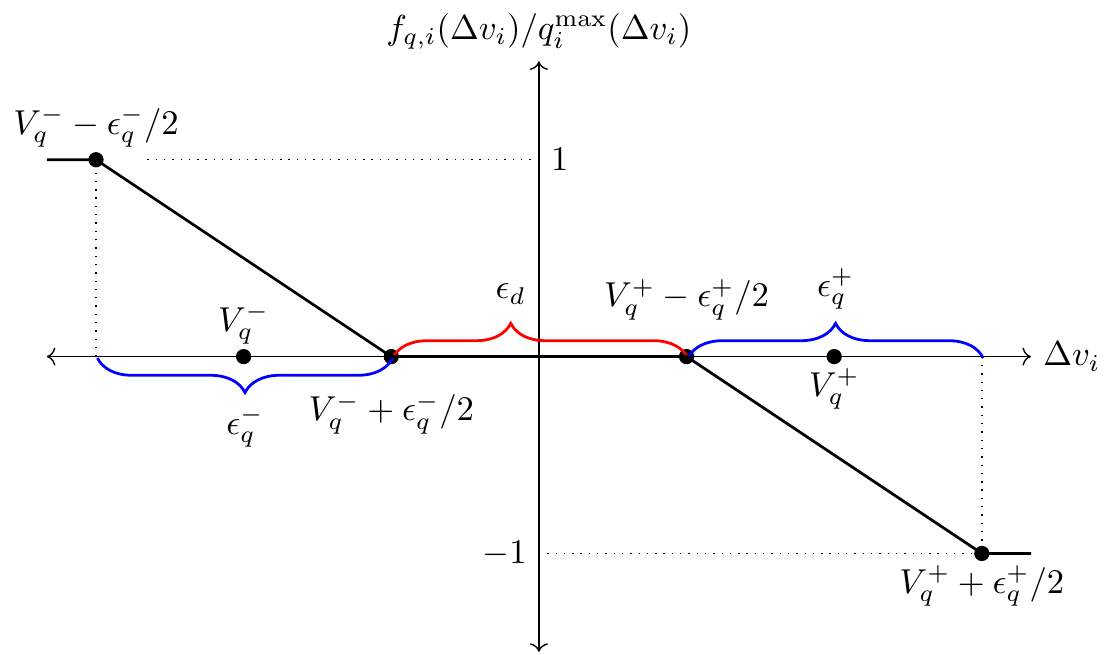}
	\caption{Inverter Volt-Var Curve}
	\label{fig:vvc_new}
\end{figure}
\begin{align}
f_{p,i}(\Delta v_i) = 
&\begin{cases}
\overline{p_i} &\Delta v_i \leq (V_p - \frac{\epsilon_p}{2}) \\[0.75em]
\dfrac{\overline{p_i}}{-\epsilon_p} \left(\Delta v_i - (V_p + \frac{\epsilon_p}{2})\right) &   
|\Delta v_i - V_p|<\frac{\epsilon_p}{2}
\\[0.75em]
0 & \Delta v_i > (V_p + \frac{\epsilon_p}{2})  \\
\end{cases} 
\label{eq:vwcsg}
\end{align}
\textcolor{blue}{
\begin{align}
& f_{q,i} (\Delta v_i) = \nonumber \\
&\begin{cases}
\overline{q_i} & \Delta v_i<V_{q}^{-}-\frac{\epsilon_q^{-}}{2}\\[0.75em]
\dfrac{\overline{q_i}}{-\epsilon_q^{-}} \Big(\Delta v_i  - (V_q^{-} + \frac{\epsilon_q^{-}}{2}) \Big) &
 \abs{\Delta v_i - V_q^{-}}<\frac{\epsilon_q^{-}}{2}
\\[1em]
0 &  |\Delta v_i|< \frac{\epsilon_d}{2}  
\\[0.75em]
\dfrac{\overline{q_i}(\Delta v_i)}{-\epsilon_q^{+}} \left(\Delta v_i  - (V_q^{+} - \frac{\epsilon_q^{+}}{2}) \right) &
 |\Delta v_i - V_q^{+}|<\frac{\epsilon_q^+}{2}
\\[1em]
-\overline{q_i} & \Delta v_i > V_{q}^{+} - \frac{\epsilon_q^{+}}{2}
\end{cases}
\label{eq:vvcsqg}
\end{align}
}

Recalling assumptions A1 and A2, the \textcolor{blue}{following} lemma confirms that the IEEE 1547 policies described above are \lipschitz functions. 
\begin{lemma}
\label{lemma.Lipschitz}
The functions $f_{p,i}(\Delta v_i)$ and $f_{q,i}(\Delta v_i)$ can be represented as \lipschitz functions where the \lipschitz constants $C_{p,i}$ and $C_{q,i}$ can be written as: 
\begin{align}
    C_{p,i} &= \dfrac{\overline{p_i}}{\epsilon_p};~~~~
    \label{eq:def_cq}
    C_{q,i} = \dfrac{\overline{p_i}}{\epsilon_p \sqrt{\frac{1}{\mu^2}-1}}+ \dfrac{q_i^{lim}}{\epsilon_q^{+}}
\end{align} 
where, $\mu$ is defined in \eqref{eq:mu}.
\end{lemma}
\begin{proof}
Taking the derivative of \eqref{eq:vwcsg},
\begin{align}
\frac{d f_{p,i}(\Delta v_i)}{d\Delta v_i} = 
&\begin{cases}-
\dfrac{\overline{p_i}}{\epsilon_p} &\mbox{if: }   (\Delta v_i - V_p) \in\  \interval[open left]{-\frac{\epsilon_p}{2}}{\frac{\epsilon_p}{2}} \\[0.75em]
0 &\mbox{if: otherwise } \\ 
\end{cases} 
\label{eq:vpwcsg}
\end{align}
Hence, the \lipschitz constant for the Volt-Watt curve of the inverter at $i^{th}$ bus is $C_{p,i}=\dfrac{\overline{p_i}}{\epsilon_p}$, \textcolor{blue}{as reported in \eqref{eq:def_cq}.}\\
Similarly, taking the derivative of \eqref{eq:vvcsqg}:
\begin{align}
&\frac{d f_{q,i}(\Delta v_i)}{d\Delta v_i}(\Delta v_i) = \nonumber \\
&\begin{cases}
-\dfrac{\Bar{q_i}}{\epsilon_q^{-}} ~ &\mbox{if: } (\Delta v_i - V_q^{-}) \in\  \interval[open right]{-\frac{\epsilon_q^{-}}{2}}{\frac{\epsilon_q^{-}}{2}} \\[1em]
\frac{d\left(\Bar{q_i}(\Delta v_i)
\frac{\Delta v_i  - (V_q^{+} - \frac{\epsilon_q^{+}}{2})}{-\epsilon_q^{+}} \right)}{d \Delta v_{i}} \!&\mbox{if: } (\Delta v_i - V_q^{+}) \in\  \interval[open left]{-\frac{\epsilon_q^{+}}{2}}{\frac{\epsilon_q^{+}}{2}} \\[1em]
0 &\mbox{if: otherwise } 
\end{cases}
\label{eq:vpcsqg}
\end{align}
\textcolor{blue}{For the second case of \eqref{eq:vpcsqg}: }
\begin{align}     
&\frac{d\left(\Bar{q_i}(\Delta v_i)
 \frac{\Delta v_i \! -\! (V_q^{+} \!-\! \frac{\epsilon_q^{+}}{2})}{-\epsilon_q^{+}} 
 \right)
 }
 {d \Delta v_{i}} \nonumber 
 \\
&=u\!\left(q_i^{lim}-\sqrt{s_i^2 - f_{p,i}^2 (\Delta v_i)}\right) \frac{-f_{p,i}(\Delta v_i) \frac{df_{p,i}(\Delta v_i)}{d\Delta v_i}}{\sqrt{ s_i^2 - f_{p,i}^2 (\Delta v_i)}} \nonumber \\ 
&\times \left(\Delta v_i- (V_q^+ - \frac{\epsilon_q^+}{2}) \right) 
+\min\left(q_i^{lim},\sqrt{s_i^2 - f_{p,i}^2(\Delta v_i)}\right)
\label{eq:Cqdv}
\end{align}
where $u(\bullet)$ is the step function. 
The \lipschitz constant for the Volt-Var curve can be calculated by taking the maximum of \eqref{eq:Cqdv}. In considering that $f_{p,i} (\Delta v_i)$ has a maximum value of $\overline{p_i}$, and using \eqref{eq:vpwcsg} and \eqref{eq:mu}, a bound for the derivative of the Volt-Var curve of the inverter at $i^{th}$ bus is the value reported in \eqref{eq:def_cq}.
\end{proof}
\textcolor{blue}{Some additional observations are useful for controlling the \lipschitz constants of smart inverters to ensure stability of the overall network}. First, it is desirable to inject as much real power as possible keeping $\epsilon_p>0$ as small as possible. Also, since the dead-band for the reactive power curve can be at most zero, the parameters $\epsilon_q^+$ and $\epsilon_q^-$ are such that:
\begin{equation}
\epsilon_q^+\leq 2V^+_q  \quad\text{and}\quad  \epsilon_q^- \leq -2V^-_q     
\end{equation}
\section{Stability Analysis}
\label{sec:stability_analysis}
In this work, the inverter dynamics presented in \cite{inverter2016} \textcolor{blue}{and \eqref{eq:main} are} used to derive the stability criterion of a distribution network with smart inverters. The main result of this paper is a sufficient condition to achieve voltage stability for a network with inverters that use only local information and logic satisfying assumptions A1 and A2 expressed in Section \ref{sec.smartinverter}.

Let $r_i(\bm{A}) = \sum_j \bm{A}_{ij}$, which is the row sum of the elements of matrix ${\bm A}$. \textcolor{blue}{The stability condition can be stated as the following Theorem}:
\begin{theorem}
\label{lem:stability}
Let  $\bm{v}^{*2}=\bar{\bm{v}}^2-\bm{Z}\bm{s}^{*} $ be the voltage vector that is a fixed point of the inverter dynamics. If A1 and A2 (c.f. Section \ref{sec.smartinverter}) hold, 
a sufficient condition for stability of an inverter dominated network is:
\begin{align}
    \label{eq:final_condition}
    C_{p,i}^2 + C_{q,i}^2 < \dfrac{2~v_{i}^{*2}}{r_i(\bm{Z} \bm{Z}^T)} 
\end{align}
\end{theorem}

\begin{proof}
Writing \eqref{eq:main} and dropping the superscript: 
\begin{align}
    \bm{v} &= \sqrt{\bar{\bm{v}}^2-\bm{Z}\bm{s}}
\end{align}
where $\bar{\bm{v}}^2=v_o^2\bm{1}+\bm{Z}\bm{s}^c$ and the square root operates element-wise. Let
\begin{align}
      \bm{f_S} (\bm{v}-\bm{v}_{nom})=\mqty[\bm{f}_p \\ \bm{f}_q];
      ~~ \bm{T}= \mqty [\bm{T}_p \\ \bm{T}_q] ;~~ \bm{C_s}= \mqty [\bm{C}_p \\ \bm{C}_q] \nonumber
\end{align}
where $\bm{T}_p$, $\bm{T}_q$, $\bm{C}_p$, and $\bm{C}_q$ are diagonal matrices containing the low-pass filter time constants (time delays) and inverter droop values, respectively. 

The inverter dynamics can be written as \eqref{eq:inverter_dynamics}:
\begin{subequations}
	\begin{align}
	\label{eq:inverter_dynamics}
	\bm{T} \dot{\bm{s}} = \bm{f_S} \left(\sqrt{\bar{\bm{v}}^2-\bm{Z}\bm{s}} - \bm{v}_{nom} \right) - \bm{s}
	\end{align}
	and for the fixed point $\bm{s}^{*}$: 
	\begin{align}
	\bm{0} = \bm{f_S} \left(\sqrt{\bar{\bm{v}}^2-\bm{Z}\bm{s}^{*}} - \bm{v}_{nom} \right) - \bm{s}^{*}
	\label{eq:eqdeltas}
	\end{align}
\end{subequations}
Defining a shift about the fixed point as $\bm{\delta s}= \bm{s}-\bm{s}^*$  and using \eqref{eq:eqdeltas}, the inverter dynamics can be rewritten as:
\begin{multline}
     \bm{T} \bm{\delta} \dot{\bm{s}} = \bm{f_S} \left( \sqrt{\bar{\bm{v}}^2-\bm{Z}\bm{s}^{*}-\bm{Z}\bm{\delta s}} - \bm{v}_{nom} \right) \\ - \bm{f_S} \left( \sqrt{\bar{\bm{v}}^2-\bm{Z}\bm{s}^{*}} - \bm{v}_{nom} \right) - \bm{\delta s}
     \label{eq:deltasdot}
\end{multline}
Now the stability of \eqref{eq:deltasdot} can be analyzed with  \lyapnouv analysis. \textcolor{blue}{Considering} the candidate \lyapnouv function:  $$J(\bm{\delta s}) = \frac{1}{2}\bm{\delta s}^{T}\ \bm{T}\ \bm{\delta s}$$ Taking the derivative of $J(\bm{\delta s})$ along the trajectory yields:
\begin{align}
    \label{eq:derivative}
    \dot{J}(\bm{\delta s}) =  \bm{\delta s}^{T}\ \bm{T} \bm{\delta}\dot{\bm{s}}
\end{align}
Putting the value of $\bm{T} \bm{\delta} \dot{\bm{s}}$ from \eqref{eq:deltasdot} in \eqref{eq:derivative} provides:
\begin{multline}
    \label{eq:Jdeltasdot}
     \dot{J}(\bm{\delta s}) =  \bm{\delta s}^{T} \left[\bm{f_S} \left(\sqrt{\bar{\bm{v}}^2-\bm{Z}\bm{s}^{*}-\bm{Z}\bm{\delta s}} - \bm{v}_{nom} \right)\right. \\\left. - \bm{f_S} \left(\sqrt{\bar{\bm{v}}^2-\bm{Z}\bm{s}^{*}} - \bm{v}_{nom} \right)\right] -\bm{\delta s}^{T} \bm{\delta s}
\end{multline}
Applying the Cauchy-Schwartz inequality on \eqref{eq:Jdeltasdot} ($\norm{\bullet}$ represents the 2 norm): 
\begin{multline}
    \dot{J} (\bm{\delta s}) \leq \norm{\bm{\delta s}} \big\Vert\bm{f_S} (\sqrt{\bar{\bm{v}}^2-\bm{Z}\bm{s}^{*}-\bm{Z}\bm{\delta s}} - \bm{v}_{nom} )- \\ \bm{f_S}  (\sqrt{\bar{\bm{v}}^2-\bm{Z}\bm{s}^{*}} - \bm{v}_{nom})\big\Vert - \norm{\bm{\delta s}}^2
    \label{eq:firstCS}
\end{multline}
and then applying the \lipschitz condition: 
\begin{multline}
\label{eq:lipschitze}
   \dot{J} (\bm{\delta s}) \leq \norm{\bm{\delta s}} \left\Vert \bm{C_s}  \left(\sqrt{\bar{\bm{v}}^2-\bm{Z}\bm{s}^{*}-\bm{Z}\bm{\delta s}} \right.\right. 
   \left. \left.- \sqrt{\bar{\bm{v}}^2-\bm{Z}\bm{s}^{*}}\ \right) \right\Vert \\ - \norm{\bm{\delta s}}^2
\end{multline}
Now by using $\sqrt{a\!-\!x}-\! \sqrt{a}\leq  \!-\frac{x}{2\sqrt{a}}$, with $a = \bar{\bm{v}}^2\!-\!\bm{Z}\bm{s}^{*}$ and $x= \bm{Z}\bm{\delta s}$, the right-hand side of \eqref{eq:lipschitze} can be written as:
\begin{multline}
    \norm{\bm{C_s} \left(\sqrt{\bar{\bm{v}}^2-\bm{Z}\bm{s}^{*}-\bm{Z}\bm{\delta s}} - \sqrt{\bar{\bm{v}}^2-\bm{Z}\bm{s}^{*}} \right)} \leq \\
     \norm{- \bm{C_s} \left\{ \dfrac{1}{2}   \diag^{-1}\left(\sqrt{\bar{\bm{v}}^2-\bm{Z}\bm{s}^{*}}\right) \bm{Z} \bm{\delta s} \right\}}
\end{multline}
This allows \eqref{eq:lipschitze} to be expressed as:
\textcolor{blue}{
\begin{multline}
   \dot{J} (\bm{\delta s}) \leq \norm{\bm{\delta s}} \norm{\bm{C_s} \left\{ \dfrac{1}{2}   \diag^{-1}\left(\sqrt{\bar{\bm{v}}^2-\bm{Z}\bm{s}^{*}}\right) \bm{Z} \bm{\delta s} \right\}} \\ - \norm{\bm{\delta s}}^2
    \label{eq:secondCS} 
\end{multline}
}
In recalling that $\bar{\bm{v}}^2-\bm{Z}\bm{s}^{*} = \bm{v}^{*2}$, where $\bm{v}^*$ is the fixed point voltage vector for the system with active inverters, and using the fact that the product of diagonal matrices is commutative, \eqref{eq:secondCS} can be written as:
\begin{multline}
    \label{eq:lambdaroot}
    \dot{J} (\bm{\delta s}) \leq \dfrac{\norm{\bm{\delta s}}^2}{\sqrt{2}} \lambda^{1/2}_{max} \left( \left[\bm{C}_p^2 + \bm{C}_q^2\right] \diag^{-1} (\bm{v}^{*2}) \bm{Z} \bm{Z}^T \right) \\ -  \norm{\bm{\delta s}}^2
\end{multline}
Applying \lyapnouv stability condition ($\dot{J}<0$) on \eqref{eq:lambdaroot}: 
\begin{align}
    \lambda_{\max} \left( \left[\bm{C}_p^2 + \bm{C}_q^2\right] \diag^{-1} (\bm{v}^{*2}) \bm{Z} \bm{Z}^T \right)  < 2
\end{align}
Applying the theorem $\lambda_{\max} (\bm{A}) \leq \max(r_i(\bm{A}))$ \cite{Garren1968} to \eqref{eq.stability}, the stability condition holds if: 
\begin{align}
    \label{eq.stability}
    \max\limits_i \left( r_i\left( \left[\bm{C}_p^2 + \bm{C}_q^2\right] \diag^{-1} (\bm{v}^{*2}) \bm{Z} \bm{Z}^T \right)\right) < 2
\end{align}
Considering that $\left[\bm{C}_p^2 + \bm{C}_q^2\right] \diag^{-1} (\bm{v}^{*2})$ are diagonal matrices, \textcolor{blue}{\eqref{eq:final_condition} can be readily obtained}.
\end{proof}
The criterion derived in \eqref{eq.stability} links voltage stability with all inverters' piece-wise linear droop control values \textcolor{blue}{and network parameters. Unlike other criteria} found in the literature such as \cite{Singhal2019,bakerNetwork2017}, it also ties the condition to the solution of the \DistFlow equations at the fixed point for inverter operations $\bm{v}^{*2}$. The condition in \eqref{eq:final_condition} also emphasizes that when $\bm{v}^{*2}$ is relatively small the stability bound becomes harder to satisfy.

\textcolor{blue}{Numerical observation shows that the sufficient condition for network stability is quite conservative.  However, the benefit of the simple expression in \eqref{eq:final_condition} is that set-points for inverters can be chosen using a local policy and local voltage information.} More specifically, when an inverter measures rapid fluctuations in voltage amplitude, it can react by bringing constants $\bm{C}_p$ and $\bm{C}_q$ below the bound established by \eqref{eq:final_condition} to restore stability. To do so, values of the coefficients $r_i({\bm Z \bm Z})^T$ should be known a priori but $v_{i}^{*}$ may not be known. However, a conservative value of $v_{i}^{*}$ to meet the stability condition at the time step of implementing the policy can be calculated using the voltage of the previous time step or a moving average of voltages calculated according to the time delay value as $v_{i}^{*}$. Considering the smart inverter control curve modeling presented in IEEE 1547 standard, \textcolor{blue}{the following analysis illustrates} how to change $C_{p,i}$ and $C_{q,i}$ for the $i^{\textnormal{th}}$ inverter in response to oscillations.

Using the result in Lemma \ref{lemma.Lipschitz} and relaxing the strict inequality by introducing a stability margin $\epsilon \geq \epsilon_0$, where $\epsilon_0>0$ is a desired lower bound on the stability margin, \textcolor{blue}{\eqref{eq:final_condition} can be written as}:
\begin{align}
    \label{eq:modified_final_condition}
    \left(
    \frac{\overline{p_i}}{\epsilon_{p,i}} \right)^2 + 
    \left(\frac{\overline{p_i}}{\epsilon_{p,i}
    \sqrt{\frac{1}{\mu^2}-1}}+ \frac{q_i^{lim}}{\epsilon_{q,i}^{+}}\right)^2  &\leq \eta_i - \epsilon 
\end{align}
where:
\begin{align}
    \label{eq:eta_i_defn}
    \eta_i = \dfrac{2~v_{i}^{*2}}{r_i(\bm{Z} \bm{Z}^T)}
\end{align}
By using: 
\begin{equation}
    \label{eq:reverse_variable}
    x_i= 1/\epsilon_{p,i}^{'},~~y_i=
   1/\epsilon_{q,i}^{+'}
\end{equation}
\eqref{eq:modified_final_condition} can be rearranged as:
\begin{equation}
    \label{eq:xMx}
    [x_i, y_i]{\bm M}_i[x_i, y_i]^T \leq \eta_i - \epsilon 
\end{equation}
where the entries of the matrix ${\bm M}_i$ are:
\begin{align*}
[{\bm M}_i]_{11}&=\dfrac{\overline{p_i}^2}{1-\mu^2};~~~
[{\bm M}_i]_{22}=(q_i^{lim})^2\\
[{\bm M}_i]_{12}&=[{\bm M}_i]_{21}=
\dfrac{\overline{p_i} q_i^{lim} }{\sqrt{1/\mu^2-1}}
\end{align*}
Now the minimization of $\epsilon_p$ allows an inverter to maintain the largest amount of active power injection. Hence, \textcolor{blue}{the following optimization problem can be formulated} : 
\begin{align}
    \max x_i~~\textnormal{subj. to}~ \eqref{eq:xMx},~x_i>0, y_i>1/(2V_q^+)
\end{align}
Given the various positivity constraints and the bound in $\epsilon_{q,i}^{+'}$, the following is true: 
\begin{equation}
    \left(
    \frac{\overline{p_i}}{\epsilon_{p,i}'} \right)^2 + 
    \left(\frac{\overline{p_i}}{\epsilon_{p,i}'
    \sqrt{\frac{1}{\mu^2}-1}}+ \frac{q_i^{lim}}{\epsilon_{q,i}'}\right)^2 
<\left(\dfrac{\overline{p_i}}{\epsilon_{p,i}' \sqrt{1-\mu^2}}\right)^2    
\end{equation}
where the right-hand side is \textcolor{blue}{simply the left-hand side with the term $\dfrac{q_i^{lim}}{\epsilon_{q,i}'}=0$. This forms an upper bound} because $0<\epsilon_{q,i}^{+'}\leq 2V_q^+$ is impossible. \textcolor{blue}{Therefore, $x_i$ can be chosen such that the upper-bound matches $(\eta_i - \epsilon)$. This results in choosing} $\epsilon_{p,i}'$ following \eqref{eq:epsilonp_prime} and setting $\epsilon_{q,i}^{+'}$ to match its upper limit $2V_q^+$. This shrinks the dead-band region in the Volt-Var characteristic and results in choosing $\epsilon_{p,i}'$ through \eqref{eq:epsilonp_prime}:
\begin{equation}
    \label{eq:epsilonp_prime}
   \epsilon_{p,i}'=\frac{\overline{p_i}}{\sqrt{(1-\mu^2)(\eta_i - \epsilon)}}
\end{equation}
The definition of $\eta_i$ shown in \eqref{eq:eta_i_defn} depicts that $\eta$ for a node decreases as the electrical distance increases of that node from the substation (due to increase in $r_i(\bm{Z} \bm{Z}^T)$ and less variation in $v_i^*$ compared to $r_i(\bm{Z} \bm{Z}^T)$). This suggests that the inverter furthest from the substation will have the largest $\epsilon_{p,i}^{'}$. So for multiple equal-sized inverters placed across the entire distribution network, the inverter furthest from the substation can inject more real power than others, thus providing voltage support while satisfying the stability criterion. \textcolor{blue}{Moreover, \eqref{eq:epsilonp_prime} illustrates that the local control policy for individual inverter relies on its capacity and operational parameters. There is no other independent variable that can affect the calculation of  $\epsilon_{p}'$ when an inverter experiences voltage oscillation at its terminal. Integrating the operating condition while deriving \lyapnouv analysis provides a conservative estimate to ensure that the control policy can stabilize oscillations under a given operating condition, thus ensuring the robustness of the proposed methodology. }

\section{Simulation Results}
\label{simulation_results}

\textcolor{blue}{Simulations were completed on the 85 bus radial MATPOWER test case \cite{matpower} with maximum active and reactive load being 2.571 MW and 2.622 MVar, respectively. The test case network was used to evaluate three cases as shown in Table \ref{table:cases}. All simulations were completed on a Intel(R) Xeon(R) CPU E5-$1630$ v3 $3.70$ GHz computer.}

\begin{table}[htbp]
\processtable{\textcolor{blue}{List of Use Cases} \label{table:cases}}
{\begin{tabular*}{20pc}{@{\extracolsep{\fill}}lll@{}}\toprule
\textcolor{blue}{Case 1} & \textcolor{blue}{No inverters (considered as the base case)}\\ 
\textcolor{blue}{Case 2} & \textcolor{blue}{Inverters without stabilization policy}\\ 
\textcolor{blue}{Case 3} & \textcolor{blue}{Inverters with stabilization policy}\\ 
\botrule
\end{tabular*}}{}
\end{table}

Smart inverters were randomly placed at 5 of nodes, highlighted in Fig. \ref{fig:85Bus}, with a total installed capacity of 1.05 MW. Each inverter has combined Volt-Var and Volt-Watt capability (Volt-Watt preference). For the simulation, each load and inverter was assigned load profiles and generation profiles, respectively, using data from \href{https://dataport.pecanstreet.org/}{Dataport} \cite{Helou2020}.

The profiles from Dataport used in this work are at $10$ minute resolution for one day. Generation profiles include real power generation data from multiple solar inverters. Load profiles are based on household real power consumption data. For this work, spline interpolation was conducted to create $1$ second load and generation profiles from the original data resulting in $86400$ data points in seconds. Simulations here use data taken from load and generation profiles between $41000-44600$ seconds, one hour over midday, to demonstrate algorithm operation.

The output of the spline interpolation process is based on the size of actual inverters and actual household demand recorded by Dataport. Hence the data is normalized to use with the rated capacities of the inverters and the loads in the 85 bus network. The individual generation profile is used to calculate the rated apparent power $s_i$ for the $i^{th}$ inverter at each time step of the simulation. Figure \ref{fig:profile} shows the average normalized load profile (averaged across 85 nodes) and average normalized generation profile (averaged across 5 nodes) spread across the $3600$ second simulation period (between $41000-44600$ seconds in the day). \textcolor{blue}{The graphs show data that were averaged after the normalization process, and therefore, the peak value is not $1.0$ as would be expected in a graph showing normalized profiles of each individual inverter. For cases 2 and 3, the choice of parameters is as follows:}
\begin{itemize}
    \item Time delay $T_d^i$ for the $i^{th}$ inverter considered is the minimum of the low-pass filter constants for real power $T_p^i$ and reactive power $T_q^i$; $T_d^i = \min(T_p^i,T_q^i)$. 
    \item \textcolor{blue}{The $i^{th}$ inverter executes it's control policy when voltage flicker measured by \eqref{eq:voltage_flicker} is greater than $v_i^T\textnormal{pu}$}. 
    \begin{equation}
        \label{eq:voltage_flicker}
        \textnormal{voltage flicker} (t_0,i) = \dfrac{\sum_{t=t_0-T_d^i+1}^{t=t_0} \abs{v_{t,i} - v_{t-1,i}}}{T_d^i}
    \end{equation}
    \item The $i^{th}$ inverter applies its local control policy at every $T_d^i~s$ based on the measured $\textnormal{\textit{voltage flicker}} (t,i)$  at time step $t$ following \eqref{eq:voltage_flicker}.
\end{itemize} 
\textcolor{blue}{Table \ref{tab:parameters} provides the parameters used for initializing the inverters in cases 2 and 3.}
%
\begin{table}[!h]
\processtable{\textcolor{blue}{Parameter Values for Cases 2 and 3}\label{tab:parameters}}
{\begin{tabular*}{20pc}{@{\extracolsep{\fill}}lll@{}}
\toprule
\textcolor{blue}{$V_p$} & \textcolor{blue}{$1.035$ pu} \\ 
\textcolor{blue}{$\epsilon_p$} & \textcolor{blue}{$0.03$ pu}\\ 
\textcolor{blue}{$V_q^{+}$} & \textcolor{blue}{$1.035$ pu} \\ 
\textcolor{blue}{$V_q^{-}$} & \textcolor{blue}{$0.965$ pu}\\ 
\textcolor{blue}{$\epsilon_q^{+}$} & \textcolor{blue}{$0.03$ pu}\\ 
\textcolor{blue}{$\epsilon_q^{-}$} & \textcolor{blue}{$0.03$ pu}\\ 
\textcolor{blue}{$\epsilon$} & \textcolor{blue}{$0.000001$} \\
\textcolor{blue}{$v^T$} & \textcolor{blue}{$0.01$ pu} \\
\textcolor{blue}{Minimum power factor} & \textcolor{blue}{$0.2$} \\
\botrule
\end{tabular*}}{}
\end{table}

The local control policy is now demonstrated under a period of time with solar generation intermittency and then a cyber-attack scenario. \textcolor{blue}{Figure \ref{fig:flowchart} shows the local control algorithm process implemented by each individual inverter.}
\begin{figure}[!t]
	\centering
	\includegraphics[width=1\columnwidth]{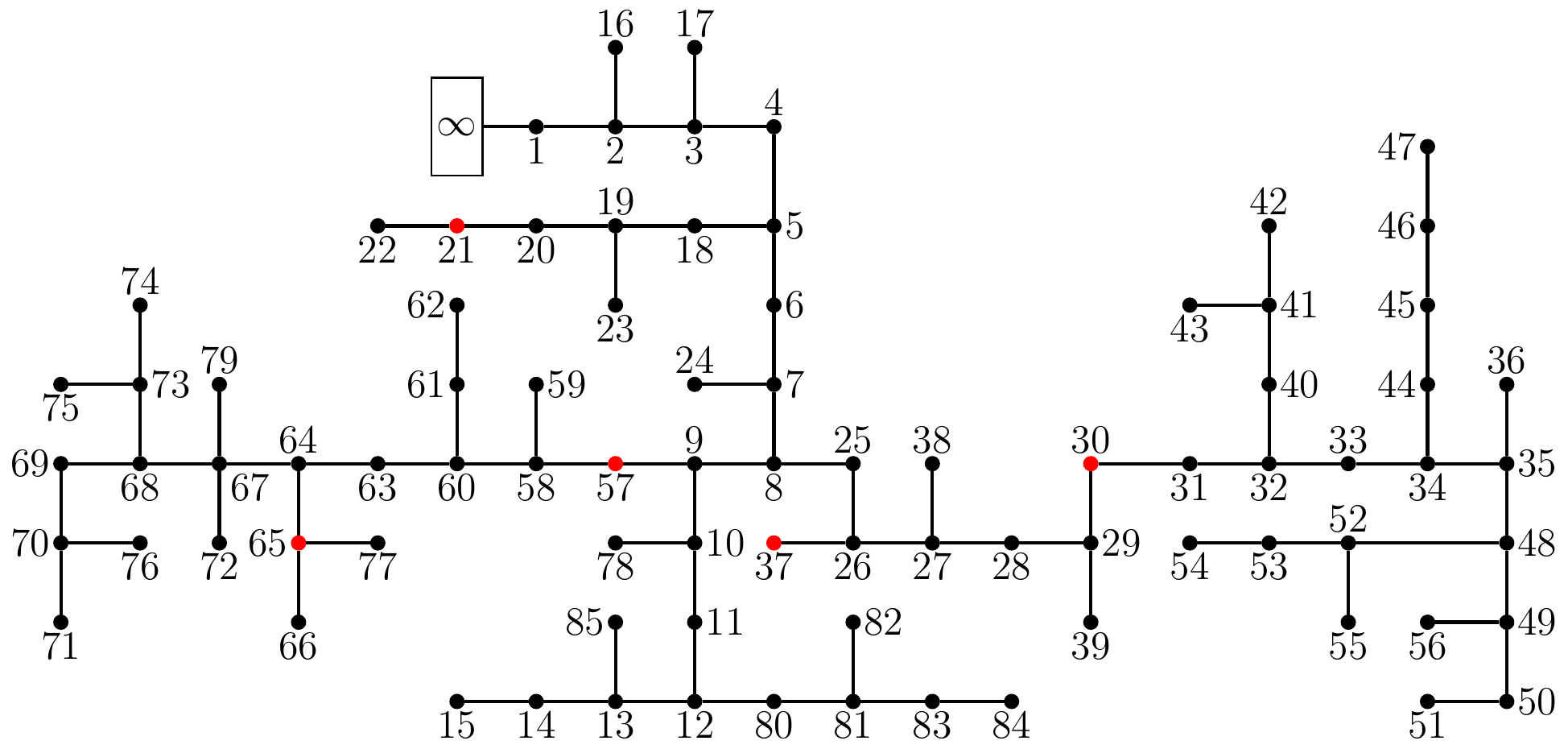}
	\caption{85 bus radial MATPOWER test case with \textcolor{red}{red} denoting location of inverters}
	\label{fig:85Bus}
\end{figure}
\begin{figure}[]
	\centering
	\includegraphics[width=1\columnwidth,height=0.25\textheight]{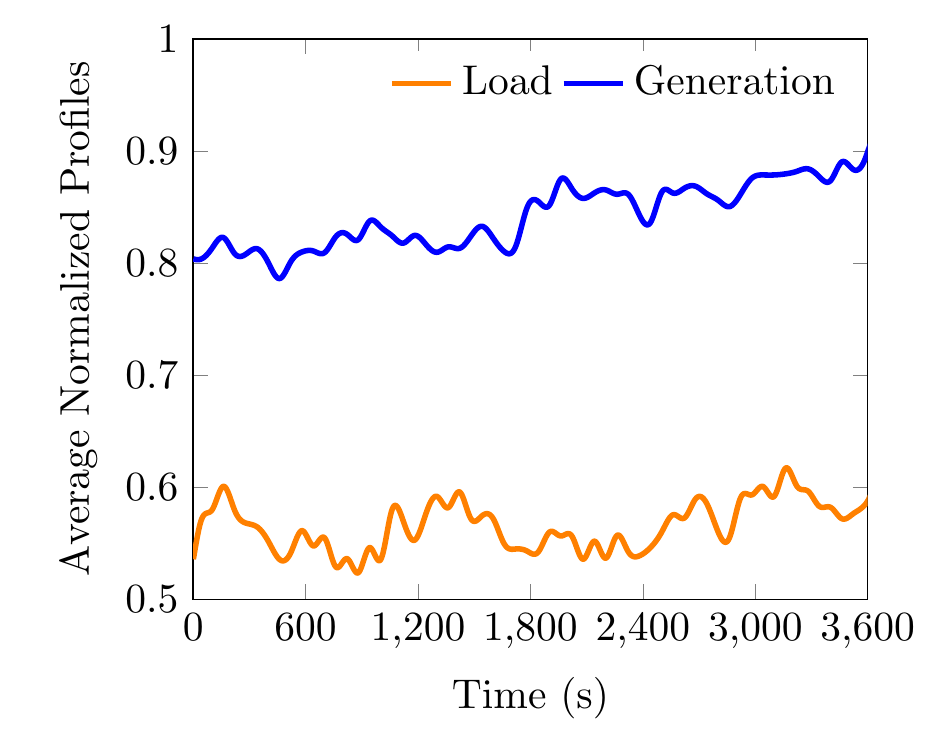}
	\caption{Normalized average load and generation profile}
	\label{fig:profile}
\end{figure}
\begin{figure}[htbp]
	\centering
	\includegraphics[scale=0.75]{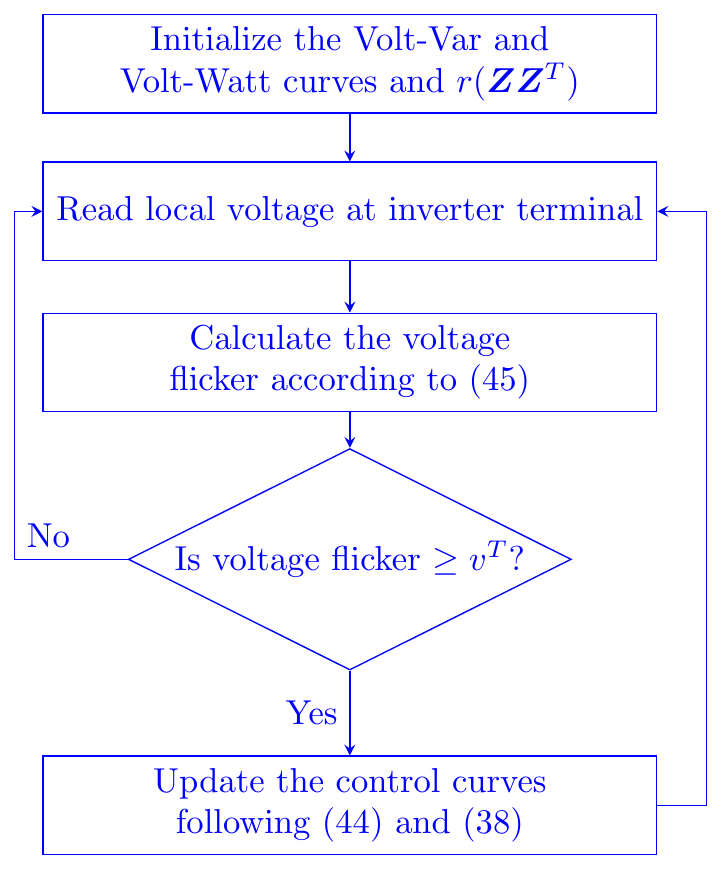}
	\caption{\textcolor{blue}{Flowchart showing the local control algorithm functions implemented by each individual inverter}}
	\label{fig:flowchart}
\end{figure}

\begin{figure*}[]
    \centering
    \includegraphics[width=1\linewidth,height = 0.35\textheight]{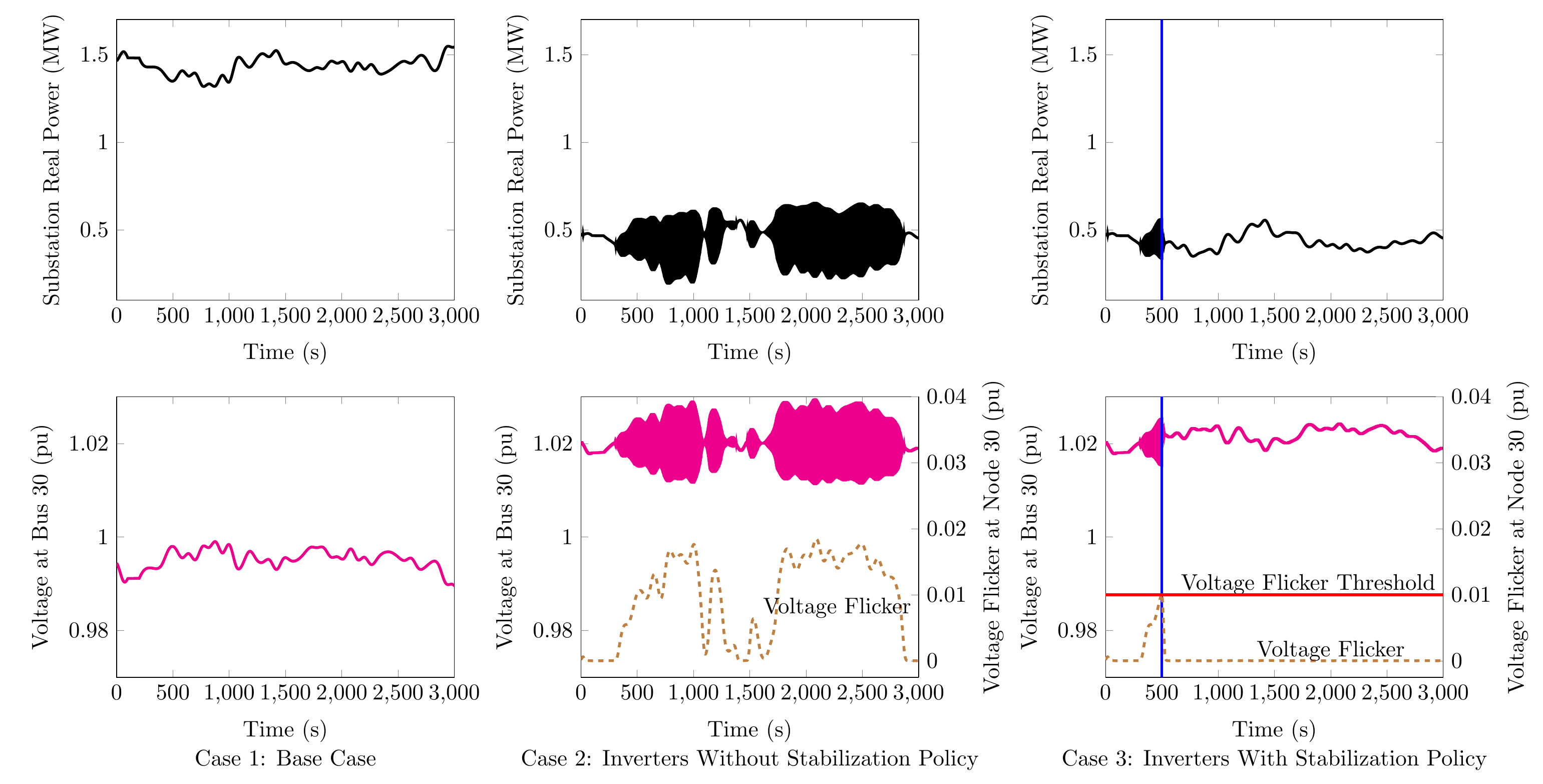}
    \caption{Substation real power (top) and Bus $30$ voltage (bottom) under three simulated cases}
    \label{fig:SubstationPlot}
\end{figure*}
\begin{figure*}[]
	\centering
	\includegraphics[width=1\textwidth,height = 0.19\textheight]{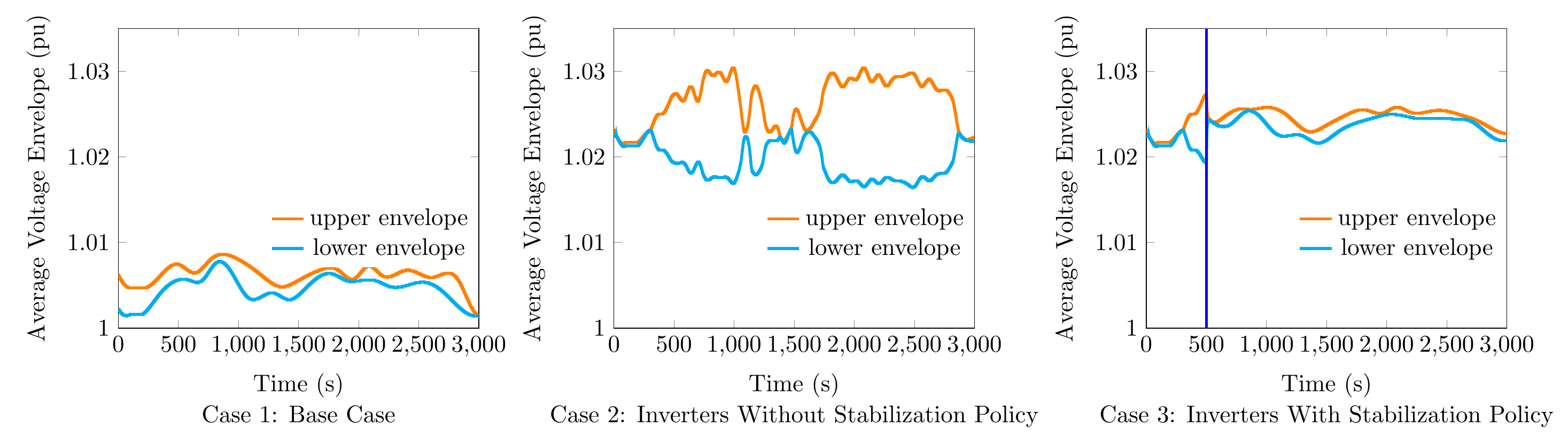}
	\caption{Average voltage envelope for all nodes under three simulated cases}
	\label{fig:voltageenvelope}
\end{figure*}
\begin{figure*}[]
	\centering
	\includegraphics[width=1\textwidth,keepaspectratio]{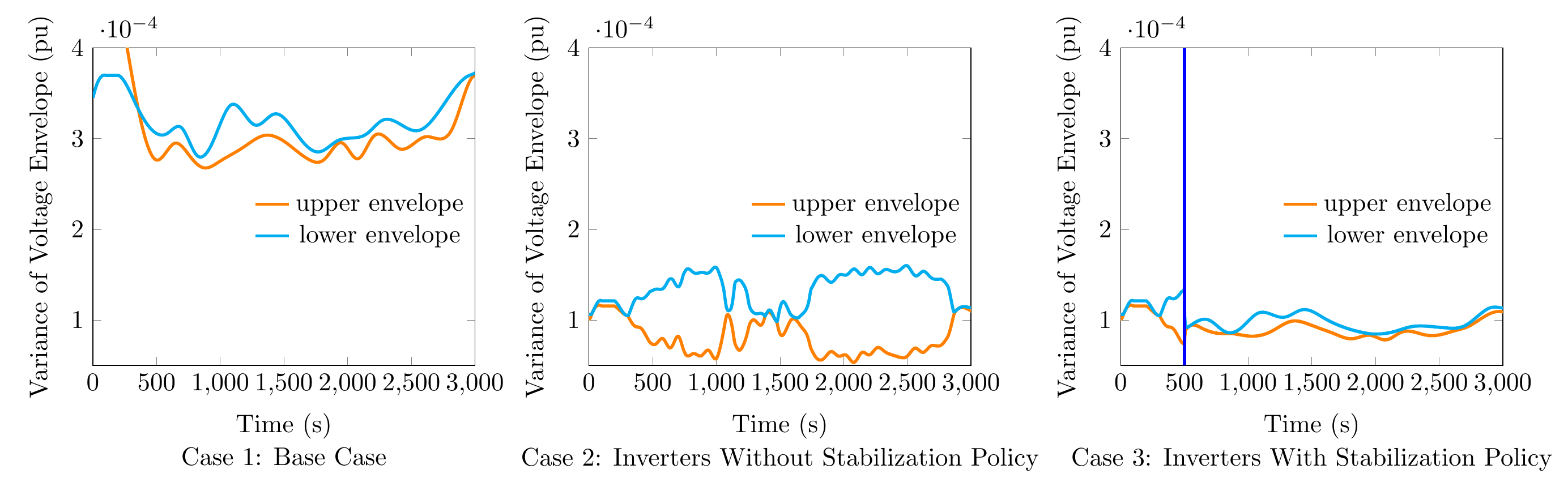}
	\caption{Variance of voltage envelope for all nodes under three simulated cases}
	\label{fig:varenvelope}
\end{figure*}
\subsection{Damping Oscillation From Generation Intermittency}
Simulations were completed with a 1 second time step for each of the three cases and the local policy is employed for every inverter. The time delay ($T_d^i$) for each inverter was chosen to be $25s$ arbitrarily. The arbitrary choice is justified by considering that the stability criterion in \eqref{eq:final_condition} does not include $\bm{T}_p$ and $\bm{T}_q$ and hence the time delay for each inverter dictates the time interval for successive implementation of the policy. The choice of an optimal time delay for inverters will be both network-specific and inverter limited, and such considerations are out of the scope of this work that demonstrates the fundamental behavior of the policy. The inverter at bus $30$ is selected to illustrate the effect of the stabilization policy. 

\begin{figure*}[!t]
    \centering
    \includegraphics[width=1\linewidth,height = 0.35\textheight]{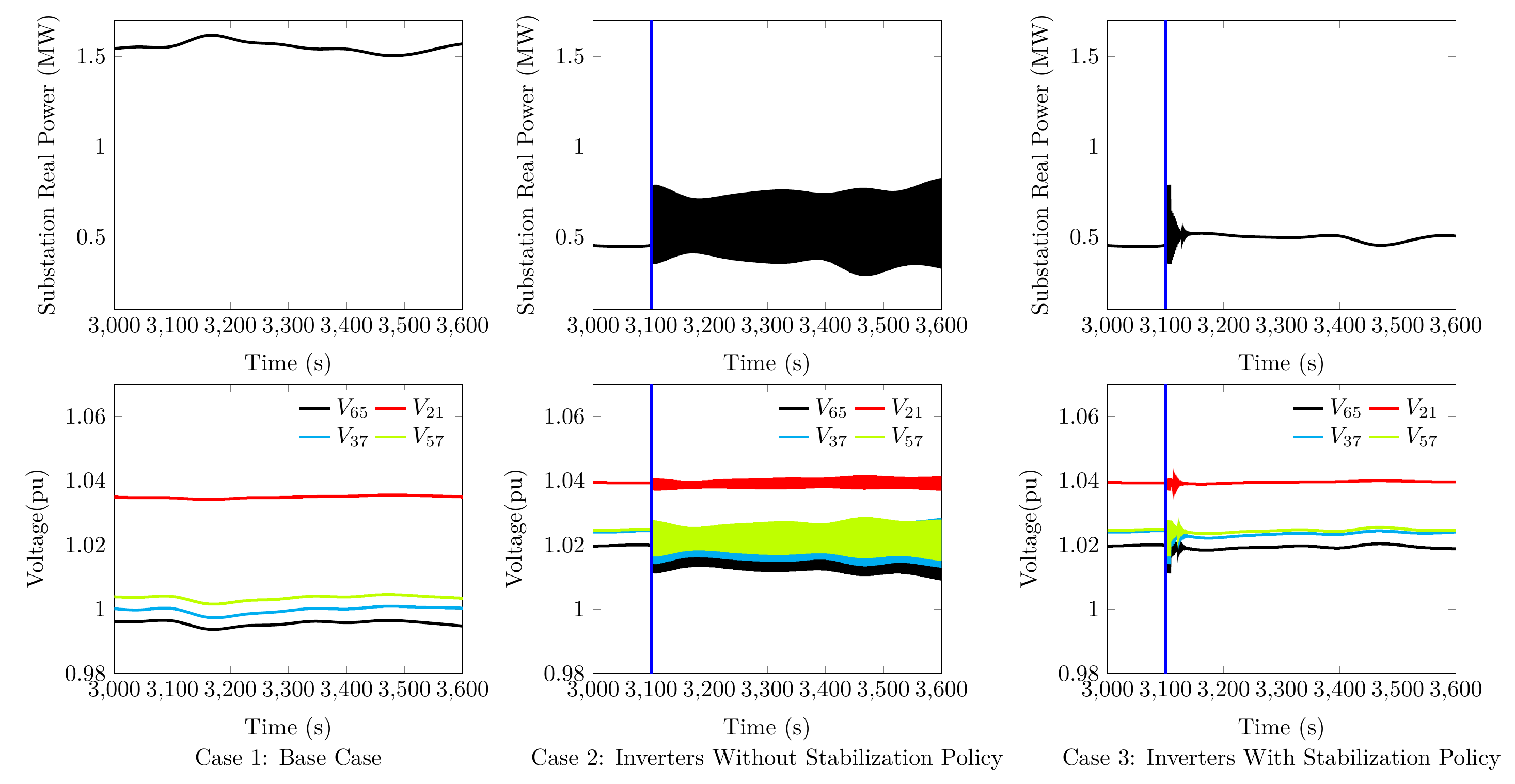}
    \caption{Substation real power (top) and unattacked inverters' bus voltages (bottom) under three simulated cases during a cyber-attack scenario}
    \label{fig:cyber_attack_plot}
\end{figure*}

Figure \ref{fig:SubstationPlot} shows the resulting voltage for the inverter at bus $30$ and real power flow at the substation. Oscillations in real power and voltage are significant when inverters are not equipped with the stabilization policy. When the stabilization policy is active and engages according to the criteria shown in \eqref{eq:voltage_flicker}, the inverter at bus $30$ adjusts its droop curve to mitigate voltage oscillation using only local voltage flicker measurements and without any centralized control or information sharing with others. Note that, though all the inverters are equipped with the policy, the flicker condition is only satisfied for the inverter at bus $30$ for this specific load and generation profile. If a tighter limit was chosen ($<0.01 ~\textnormal{pu}$), the other inverters may have been triggered to adjust their droop curves to mitigate the oscillation, or it could be that only some of the inverters are called upon to mitigate the oscillation, as shown in this example. The choice of the appropriate voltage flicker threshold can be evaluated through simulations of specific network configurations, inverter types, and network states prior to physical implementation. Case 3 of Figure \ref{fig:SubstationPlot} shows the voltage at bus $30$ and the stabilization policy activating when voltage flicker exceeds the specified threshold at $t=500s$. Oscillations are rapidly dampened and the bus voltage stabilizes. Results in Figure \ref{fig:SubstationPlot} also show that oscillations cease in substation real power flow. 

\textcolor{blue}{To illustrate that the developed policy can stabilize voltage oscillation in the network overall, the upper and lower envelope of voltages calculated across all nodes are shown in Figure \ref{fig:voltageenvelope} and Figure \ref{fig:varenvelope}. This simplifies results for graphical display rather than showing voltages for all 85 nodes.} Then the upper envelopes and lower envelopes are averaged for all the $85$ nodes for all the simulated cases and shown in Figure \ref{fig:voltageenvelope}. Variance for the upper envelope and lower envelopes for voltages across all nodes for all the cases is shown in Figure \ref{fig:varenvelope}. \textcolor{blue}{It is evident from Figure \ref{fig:voltageenvelope} that the difference between the upper and lower envelopes of inverter voltages is much lower for case 3 than for case 2. A similar trend is found for the calculated voltage variances as shown in Figure \ref{fig:varenvelope}.} These two figures provide further insight into network-level behaviors by showing that, in the scenario for inverters with a stabilization policy, the average network voltage profile closely follows the shape and variance of the scenario without inverters, whereas greater deviations exist in the scenario of inverters without a stabilization policy. This demonstrates that the proposed policy was effective in mitigating voltage instability using only local voltage information.

\subsection{Damping Oscillation From Cyber-Attack}
To illustrate that the developed policy can mitigate instability created from a cyber-attack, a scenario was generated with heterogeneous time delays for the inverters. The time delays were sampled from $\mathcal{N}(5s,10s)$ to acknowledge that inverters from different manufacturers can have different time delays. The simulation time window was selected during a period of the day with low generation intermittency that was insufficient to trigger the policy by voltage flicker. The severity of a cyber-attack is affected by which inverters are attacked, how inverter parameters are changed, and the operating state of the network. For this demonstration, the inverter at bus $30$ was chosen because bus $30$ is electrically the furthest from the substation among the buses with an inverter. This results in a higher value of $r_i(\bm{Z}\bm{Z}^T)$, hence making the bound shown in \eqref{eq:final_condition} the easiest to violate from an attacker's perspective. At $t=3100s$, a cyber-attack modifies $V_p$ and $\epsilon_p,$ of the inverter at bus $30$ to $1.02$ and $0.2$, respectively. The inverter also loses its capability to apply the local control policy while under attack. Set-point values for the inverter were selected to increase the slope of the Volt-Watt curve and push $C_p$ into violation of the stability criterion \textcolor{blue}{derived in \eqref{eq:final_condition}}. Case 2 of Figure \ref{fig:cyber_attack_plot} shows how a cyber-attack on a single inverter can disrupt bus voltages for the other inverters if they are not equipped with the stability policy. 

Case 3 shows results for inverters with the stability policy enabled. The four inverters enable the policy according to their local measurements and individual time delays and begin to dampen oscillations and finally achieve stability after approximately 25 seconds. The benefit of the policy is also shown to reduce swings and finally stabilize substation real power flow. 

\section{Discussion}
\label{sec:discussion}
This work developed and simulated an approach to enable voltage stabilization in a distribution network by updating solar PV inverter set-points in real-time \textcolor{blue}{to counteract voltage oscillation}. Piece-wise linear models of both Volt-Var and Volt-Watt functions were represented using \lipschitz functions. \textcolor{blue}{\lipschitz constants for both Volt-Var and Volt-Watt control functions were derived while respecting smart inverter hardware limitations for reactive power generation. \lyapnouv analysis was used to derive a sufficient condition to ensure network-wide voltage stability. Using the condition, a local control policy to adjust the droop constants of control curves was derived that used only local information.} Simulation results completed for a radial distribution test case network were demonstrated to show how inverters can adapt their control curves to voltage oscillations caused by solar intermittency and cyber-attacks.

\textcolor{blue}{Implementation of the proposed algorithm in a real distribution network relies on the validity of the assumptions mentioned in Section \ref{sec:inverter_modeling}. The assumptions hold in any practical application due to: 1) Monotonically increasing functions for an inverter increase real power output when the voltage is beyond the threshold voltage and absorbs or injects reactive power when the voltage is below or above nominal voltage, respectively, and 2) An unbounded derivative for control curves requires a step change rather than a ramp change in smart inverter real or reactive power production, which violates physical constraints of the inverter. In thinking of implementing the local control policy on physical inverters, these controls can be integrated directly by inverter vendors upon manufacture, or a legacy inverter can be equipped with a secondary controller (e.g., Raspberry Pi) that can implement the control curves and update operating set-points.}

Future work will explore how this approach can counteract voltage instabilities for networks including the smart inverter capabilities introduced here and other voltage regulating devices such as voltage regulators and capacitor banks. \textcolor{blue}{Future work will also evaluate the derived policy in a power hardware-in-the-loop simulation including a real-time digital simulator, grid emulator, programmable load bank, and smart inverter(s).}
\section{Acknowledgments}
\label{sec11}
This research was supported in part by the Director, Cybersecurity, Energy Security, 
and Emergency Response, Cybersecurity for Energy Delivery Systems program of the U.S. Department of Energy via the Cybersecurity via Inverter-Grid Automatic Reconfiguration (CIGAR) project under contract DE-AC02-05CH11231, the United States Office of Naval Research under the Defense University Research-to-Adoption (DURA) Initiative with Award Number N00014-18-1-2393, and the Engineering Research Center Program of the National Science Foundation and the Office of Energy Efficiency and Renewable Energy of the Department of Energy under NSF Cooperative Agreement Number EEC-1041895. Any opinions, findings and conclusions or recommendations expressed in this material are those of the authors and do not necessarily reflect those of the Office of Naval Research, National Science Foundation or Department of Energy.
\bibliographystyle{IEEEtran}
\bibliography{Bib_InverterJournal}
\end{document}